\documentclass[letterpaper]{IEEEtran}

\usepackage{graphicx,epsfig}
\usepackage[noadjust]{cite}
\usepackage{mcite}
\usepackage{amsfonts,helvet}
\usepackage{fancyhdr}
\usepackage{threeparttable}
\usepackage{epsf,epsfig}
\usepackage{amsthm}
\usepackage{amsmath}
\usepackage{siunitx}
\usepackage{amssymb}
\usepackage{dsfont}
\usepackage{subfigure}
\usepackage{color}
\usepackage{algorithm}
\usepackage{algpseudocode}
\usepackage{algcompatible}
\usepackage{enumerate}
\usepackage{cancel}
\usepackage{lipsum,amsmath,multicol}
\usepackage{cuted}

\newtheorem{theorem}{Theorem}

\newtheorem{lemma}{Lemma}

\newtheorem{proposition}{Proposition}

\usepackage{eucal}

\begin{document}

\title{Rate Analysis of Ultra-Reliable Low-Latency Communications in Random Wireless Networks}


\author{Jeonghun~Park
\thanks{J. Park is with the School of Electronics Engineering, College of IT Engineering, Kyungpook National University, Daegu, 41566, South Korea (e-mail: jeonghun.park@knu.ac.kr).

This work was supported by Electronics and Telecommunications Research Institute (ETRI) grant funded by the Korean government [2019-0-00964, Development of Incumbent Radio Stations Protection and Frequency Sharing Technology through Spectrum Challenge]. 
}}

\maketitle \setcounter{page}{1} 

\begin{abstract}
In this letter, we analyze the achievable rate of ultra-reliable low-latency communications (URLLC) in a randomly modeled wireless network. We use two mathematical tools to properly characterize the considered system: i) stochastic geometry to model spatial locations of the transmitters in a network, and ii) finite block-length analysis to reflect the features of the short-packets. Exploiting these tools, we derive an integral-form expression of the decoding error probability as a function of the target rate, the path-loss exponent, the communication range, the density, and the channel coding length. We also obtain a tight approximation as a closed-form. The main finding from the analytical results is that, in URLLC, increasing the signal-to-interference ratio (SIR) brings significant improvement of the rate performance compared to increasing the channel coding length. Via simulations, we show that fractional frequency reuse improves the area spectral efficiency by reducing the amount of mutual interference. 
\end{abstract}

\section{Introduction}
As new applications based on machine-to-machine communications have emerged, latency and reliability have became keys to support those applications.  
For example, in vehicle-to-everything (V2X) communications or industrial automation, a particular set of messages, e.g., safety messages, should be delivered within very low latency, say less than a few milliseconds, while achieving very high reliability, say more than $99.999\%$ \cite{osseiran:commag:14}. 
Motivated by this, a new design of communication systems so called ultra-reliable low-latency communications (URLLC) has been gaining attention.

In URLLC, the packet size becomes small to reduce the end-to-end latency, and this results in decreasing of the channel coding length \cite{durisi:procdieee:16}. In this regime, the classical Shannon capacity is not appropriate to characterize the achievable rate since the non-negligible gap occurs between the Shannon capacity and the actual rate. 
To take this gap into account, in \cite{poly:tit:10, yang:tit:14}, the achievable rate in a finite block-length regime was derived as a function of the decoding error probability, the channel coding length, and the signal-to-noise ratio (SNR). 
Leveraging this result, the rate performance of URLLC has been actively investigated in prior work. 
In \cite{schiessl:tcom:18}, the decoding error probability was analyzed in imperfect channel state information, and the optimum training sequence length to minimize the delay violation probability was found. 
This result was extended by incorporating a multi-user multiple-input single-output (MISO) system in \cite{schiessl:jsac:19}. 
In \cite{makki:wcl:14}, the achievable rate by using incremental redundancy hybrid automatic repeat request (HARQ) was analyzed in the finite block-length regime. 
In \cite{makki:wcl:19}, the MIMO achievable rate in the finite block-length regime was studied. 

One common assumption of the aforementioned prior work is that a link-level perspective is employed, where a few communication nodes were assumed and their spatial locations are deterministic. 
This perspective, however, is limited to obtain network-wise insights by averaging the spatial locations of the nodes in large-scale wireless networks. 
Stochastic geometry is a useful tool to investigate the network-wise characteristics, yet the prior work on stochastic geometry \cite{andrews:commag:10} implicitly assumed long channel coding length. As explained, this is not desirable to analyze URLLC. 

In this letter, we analyze the rate performance of URLLC in a random wireless network. 
Specifically, we consider a spectrum sharing wireless network, where each transmitter equipped with a single antenna is distributed as a homogeneous Poisson point process (PPP). Each transmitter sends data to an associated receiver equipped with a single antenna.
We assume that the number of channel use is in an order of ${\sim}100$, thereby the considered system is in the finite block-length regime. 
By jointly exploiting the tools from stochastic geometry and finite block-length regime analysis, we obtain an integral expression of the decoding error probability as a function of i) the target rate, ii) the path-loss exponent, iii) the communication range, iv) the channel coding length, and v) the density. 
To provide more insights, we also derive a tight approximation of the decoding error probability as a closed form. We verify the accuracy of the obtained approximation numerically. 
The major finding from our analytical results is that, in URLLC, increasing the signal-to-interference ratio (SIR) brings significant improvement of the rate performance compared to increasing the channel coding length. 
It is also shown that fractional frequency reuse brings significant rate performance gains by mitigating the amount of interference. 

\section{System Model}

\subsection{Network Model}
A spectrum-sharing wireless ad-hoc network is considered. 
We assume that each transmit and receiver node is equipped with a single antenna.
The locations of the transmitters, denoted as $\{{\bf{d}}_i, i\in \mathbb{N} \}$, are distributed as a homogeneous PPP with density $\lambda$. We write the set of the locations as $\Phi = \{{\bf{d}}_i, i \in \mathbb{N}\}$. 
The receiver located at ${\bf{d}}^{\rm o}_{i}$ is associated with the transmitter located at ${\bf{d}}_i$, so that the transmitter and the receiver pair is located at $\left({\bf{d}}_i,{\bf{d}}_i^{\rm o}\right)$. 

Per Slivnyak's theorem \cite{baccelli:inria}, we focus on the typical receiver located at ${\bf{d}}_1^{\rm o} = {\bf{0}}$. The received signal at the typical receiver is described as
\begin{align}
y_1 = \left\| {\bf{d}}_1 \right\|^{-\beta/2} h_1 s_1 + \sum_{i = 2}^{\infty} \left\| {\bf{d}}_i \right\|^{-\beta/2}h_i s_i + n_1,
\end{align}
where $h_i$ is a fading coefficient between the transmitter at ${\bf{d}}_i$ and the typical receiver, $s_i$ is a transmit signal sent from the transmitter at ${\bf{d}}_i$, and $n_1 \sim \mathcal{CN}(0, 1)$ is the additive Gaussian noise. We consider Rayleigh fading, so that $h_i \sim \mathcal{CN}(0,1)$. 
Assuming that the transmit power is $P$, the SINR of the typical receiver, denoted as $\gamma$, is defined as
\begin{align}
 \gamma = \frac{\left\| {\bf{d}}_1 \right\|^{-\beta} |h_1|^2 }{\sum_{i = 2}^{\infty} \left\| {\bf{d}}_i \right\|^{-\beta} |h_i|^2 + \frac{1}{P}}
\end{align}

\subsection{Finite Block-length Regime}

As shown in \cite{poly:tit:10}, the achievable rate when the decoding error probability is $\varepsilon$ and the channel coding length is $n$ is well approximated as
\begin{align} \label{eq:r_fbl}
R \approx \log_2\left(1 + \gamma \right) - \sqrt{\frac{\CMcal{V}_{}(\gamma)}{n}} Q^{-1} \left(\varepsilon	\right)
\end{align}
where $\CMcal{V}_{}(\cdot)$ is the channel dispersion determined by a channel condition. For example, in AWGN channels, this channel dispersion is given as
\begin{align} \label{eq:ch_disper_awgn}
\CMcal{V}(\gamma) = \CMcal{V}_{\sf AWGN}(\gamma) = \left(1 - \frac{1}{(1 + \gamma)^2} \right) \log_2^2(e).
\end{align}
To achieve \eqref{eq:r_fbl} with $\CMcal{V}_{\sf AWGN}(\gamma)$, the transmitter needs to use a non-Gaussian codebook \cite{poly:tit:10}. 
Treating interference as noise, each receiver experiences non-Gaussian interference, so that the corresponding channel dispersion is not equivalent to AWGN. For this reason, we cannot use \eqref{eq:ch_disper_awgn}. 
Instead, we use the results of \cite{scarlett:tit:17}. In \cite{scarlett:tit:17}, assuming that an interference channel where each link uses iid Gaussian codebook and employs nearest-neighbor decoding, the channel dispersion was derived as 
\begin{align} \label{eq:ch_disper_iid}
\CMcal{V}_{\sf iid} (\gamma) = \frac{2  \gamma}{1 + \gamma} \log_2^2(e). 
\end{align}
We note that even though the considered setting, i.e., using point-to-point iid Gaussian codebook in an interference channel, is not optimal, but it sheds light on performance insights in a practical setup.  


We assume a quasi-static channel condition, where channel coefficients (including long-term and short-term coefficients) keep constant during the whole communication process. This makes sense particularly in URLLC since it spans only small transmission slots.  

\subsection{Performance Metrics}


Under the assumed system setups, the decoding error probability is tightly approximated as \cite{scarlett:tit:17}
\begin{align} \label{eq:error_int}
\varepsilon \approx \mathbb{E}_{\gamma} \left[ Q \left( \frac{\log_2 \left( 1+ \gamma \right) - R}{\sqrt{\CMcal{V}_{\sf iid}(\gamma)/n}} \right) \right],
\end{align}
where the expectation is taken over $\gamma$. We note that there exist two randomnesses that determines $\gamma$: i) the interfering nodes' locations $\Phi \backslash {\bf{d}}_1$, ii) the short-term fading coefficients $h_i$ for $i \in \mathbb{N}$. 

Given the decoding error probability $\varepsilon$, we define $R_{\varepsilon}$ as the maximum rate that the corresponding error probability is smaller than $\varepsilon$. This is formally defined as
\begin{align}
R_{\varepsilon} = {\arg \max}_{R} \left\{\mathbb{E}_{\gamma} \left[ Q \left( \frac{\log_2 \left( 1+ \gamma \right) - R}{\sqrt{\CMcal{V}_{\sf iid}(\gamma)/n}} \right) \right] < \varepsilon \right\}. 
\end{align}
With $R_{\varepsilon}$ and $\varepsilon$, the area spectral efficiency measures the achievable rate per channel use in a unit area. This is defined as $\lambda R_{\varepsilon}(1-\varepsilon)$. 






\section{Analytical Results}

In this section, we present the main analytical results of this letter. 
We first obtain the probability density function (PDF) of $\gamma$ in the following lemma. 
\begin{lemma} \label{lem:pdf}
In a network modeled by a homogeneous PPP $\Phi$, the PDF of $\gamma$ is obtained as
\begin{align}
f_{\gamma} (x) = &\exp\left(-\pi \lambda x^{2/\beta} D^2 \frac{2\pi / \beta}{\sin \left(2\pi / \beta \right) } - x \frac{D^{\beta}}{P} \right)\cdot \nonumber \\
& \left( \frac{D^{\beta}}{P} + \frac{2}{\beta} \pi \lambda D^2 \frac{2\pi / \beta}{\sin \left(2\pi / \beta \right) } x^{-1 + \frac{2}{\beta}} \right),
\end{align}
where $D = \left\|{\bf{d}}_1 \right\|$, i.e., the link distance between the typical receiver and the corresponding transmitter. 
 \end{lemma}
\begin{proof}
It is well known that the cumulative distribution function (CDF) of $\gamma$ is obtained as 
\begin{align} \label{eq:cdf}
\mathbb{P}\left[\gamma < x \right] = 1- \exp\left(-\pi \lambda x^{2/\beta} D^2 \frac{2\pi / \beta}{\sin \left(2\pi / \beta \right) } - x \frac{D^{\beta}}{P} \right),
\end{align}
where $D = \left\|{\bf{d}}_1 \right\|$. Differentiating \eqref{eq:cdf} with regard to $x$, we have
\begin{align}
&\frac{\partial \mathbb{P}\left[\gamma < x \right]}{\partial x} =  f_{\gamma} (x) = \nonumber \\
&\exp\left(-\pi \lambda x^{2/\beta} D^2 \frac{2\pi / \beta}{\sin \left(2\pi / \beta \right) } - x \frac{D^{\beta}}{P} \right)\cdot \nonumber \\
&  \left( \frac{D^{\beta}}{P} + \frac{2}{\beta} \pi \lambda D^2 \frac{2\pi / \beta}{\sin \left(2\pi / \beta \right) } x^{-1 + \frac{2}{\beta}} \right).  \nonumber 
\end{align}
This completes the proof. 
\end{proof}



With Lemma \ref{lem:pdf}, we obtain the decoding error probability $\varepsilon$ as follows. 

\begin{theorem} \label{thm:error_prob}
In a network modeled by a homogeneous PPP $\Phi$, the decoding error probability is 
\begin{align} \label{eq:error_prob_int}
\varepsilon \approx \int_{0}^{\infty} &Q \left( \frac{\log_2 \left( 1+ x \right) - R}{\sqrt{V_{\sf iid}(x)/n}} \right) e^{\left(-\pi \lambda x^{2/\beta} D^2 \frac{2\pi / \beta}{\sin \left(2\pi / \beta \right) } - x \frac{D^{\beta}}{P} \right)}\cdot \nonumber \\
& \left( \frac{D^{\beta}}{P} + \frac{2}{\beta} \pi \lambda D^2 \frac{2\pi / \beta}{\sin \left(2\pi / \beta \right) } x^{-1 + \frac{2}{\beta}} \right)  {\rm d} x.
\end{align}
\end{theorem}
\begin{proof}  
The proof is straightforward by computing integration \eqref{eq:error_int} with regard to the PDF $f_{\gamma}(x)$. 
\end{proof}

\begin{figure}[!t]
\centerline{\resizebox{0.8\columnwidth}{!}{\includegraphics{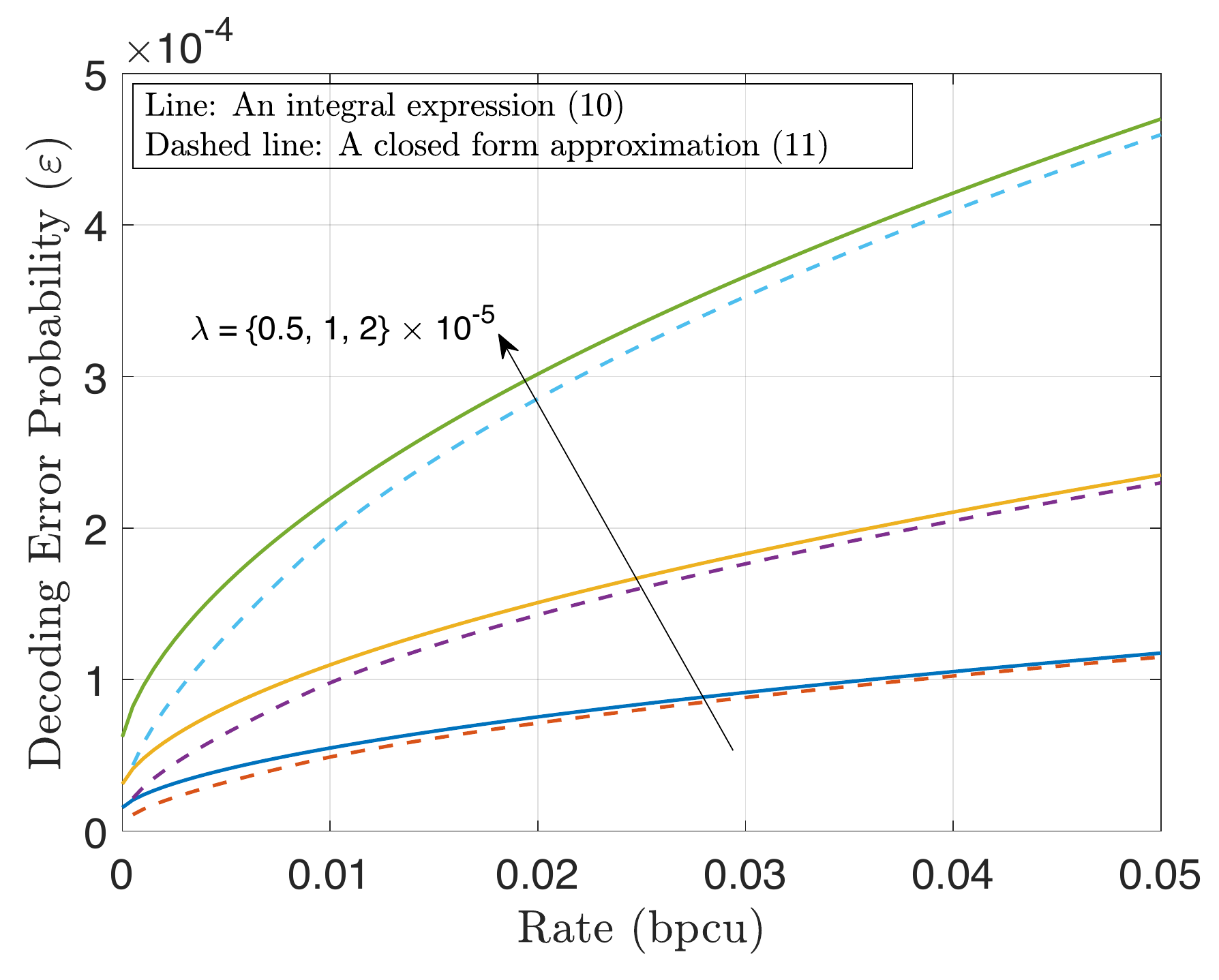}}}     
\caption{A verification plot of a closed-form approximation \eqref{eq:epsilon_approx_total}. The simulation setups are as follows: $\lambda = \{0.5, 1, 2\} \times 10^{-5}$, $n = 500$, $D = 5$, and $\beta = 4$. }
 \label{fig:verify}
\end{figure} 

Unfortunately, a closed form of \eqref{eq:error_prob_int} does not exist. 
For providing more insights, we approximate Theorem \ref{thm:error_prob}. 
We first assume that the typical receiver is interference-limited, i.e., $1/P \rightarrow 0$. We present the proposed approximation in the following.

\begin{figure*}[!]
\begin{align} \label{eq:epsilon_approx_total}
\varepsilon \approx & \;1 + 
\left(\frac{1}{2} + \sqrt{\frac{n (2^R - 1)}{4\pi 2^R}} \right) \cdot  e^{\left(- \pi \lambda D^2 \frac{2\pi/\beta}{\sin (\pi/\beta)} \left(2^{R}-1+\sqrt{\frac{{\pi (2^{2R} - 2^R)}}{{n}}}\right)^{\frac{2}{\beta}} \right)}   \nonumber \\
&  -\left(\frac{3}{2} + \sqrt{\frac{n (2^R - 1)}{4\pi 2^R}} \right) \cdot   e^{\left(- \pi \lambda D^2 \frac{2\pi/\beta}{\sin (\pi/\beta)} \left(2^{R}-1-\sqrt{\frac{{\pi (2^{2R} - 2^R)}}{{n}}}\right)^{\frac{2}{\beta}} \right)} - \left( \frac{2}{\beta}\sqrt{\frac{n}{4\pi (2^{2R} - 2^R)}} \pi \lambda D^2 \frac{2\pi/\beta}{\sin (2\pi/\beta)} \right) \cdot \nonumber \\
&  \left\{\left( 2^{R}-1-\sqrt{\frac{{\pi (2^{2R} - 2^R)}}{{n}}} \right) ^{1+\frac{\beta}{2}} E_{-\beta/2}\left(\left( 2^{R}-1-\sqrt{\frac{{\pi (2^{2R} - 2^R)}}{{n}}} \right) \pi \lambda D^2 \frac{2\pi/\beta}{\sin (2\pi/\beta)}\right)  \right. \nonumber \\
& \left.- \left(2^{R}-1+\sqrt{\frac{{\pi (2^{2R} - 2^R)}}{{n}}} \right)^{1 + \beta/2}E_{-\beta/2}\left(\left( 2^{R}-1+\sqrt{\frac{{\pi (2^{2R} - 2^R)}}{{n}}} \right) \pi \lambda D^2 \frac{2\pi/\beta}{\sin (2\pi/\beta)}\right) \right\}.
\end{align}
\hrulefill
\end{figure*}	

\begin{proposition} \label{prop:approx}
Assuming that the typical receiver is interference-limited, i.e., $1/P \rightarrow 0$, the decoding error probability $\varepsilon$ is approximated as in \eqref{eq:epsilon_approx_total}. 
\end{proposition}
\begin{proof}
To obtain a closed-form expression, we first approximate the $Q$-function as a piecewise linear function as follows.
\begin{align} \label{eq:Q_taylor}
&Q\left( \frac{\log_2 \left( 1+ x \right) - R}{\sqrt{V_{\sf iid}(x)/n}} \right) \approx  \nonumber \\
& \left\{\begin{array}{ll} {1,} & {x \le A,} \\ {\frac{1}{2} - \frac{\sqrt{n}}{\sqrt{4 \pi \cdot(2^{2R}-2^{R})}} (x-(2^{R}-1)),} & {A \le x \le B,} \\ {0,} & {x \ge B,}   \end{array} \right. 
\end{align}
where $A = 2^{R}-1-\sqrt{\frac{{\pi (2^{2R} - 2^R)}}{{n}}}$ and $B = 2^{R}-1+\sqrt{\frac{{\pi (2^{2R} - 2^R)}}{{n}}}$. The approximation in the range of $A \le x \le B$ comes from the Taylor approximation at $2^{R}-1$. 
With the approximation \eqref{eq:Q_taylor}, we calculate the integration \eqref{eq:error_prob_int} as 
\begin{align} \label{eq:epsilon_approx}
& \varepsilon \approx  \int_{0}^{A} e^{\left(-\pi \lambda x^{2/\beta} D^2 \frac{2\pi/\beta}{\sin (2\pi/ \beta)} \right)} \left( \frac{2}{\beta} \pi \lambda D^2 \frac{2\pi/\beta}{\sin (2\pi/ \beta )} x^{-1+2/\beta} \right)  \nonumber \\
&+  \int_{A}^{B} e^{\left(-\pi \lambda x^{2/\beta} D^2 \frac{2\pi/\beta}{\sin (2\pi/ \beta)} \right)} \left( \frac{2}{\beta} \pi \lambda D^2 \frac{2\pi/\beta}{\sin (2\pi/ \beta )} x^{-1+2/\beta} \right) \cdot \nonumber \\
&   \left( \frac{1}{2} - \frac{\sqrt{n}}{\sqrt{4 \pi \cdot(2^{2R}-2^{R})}} (x-(2^{R}-1))\right) {\rm d} x.
\end{align}

The first term of \eqref{eq:epsilon_approx} is derived as
\begin{align} \label{eq:epsilon_approx_pt1}
1 -  e^{\left(- \pi \lambda D^2 \frac{2\pi/\beta}{\sin (\pi/\beta)} \left(A\right)^{2/\beta} \right)}.
\end{align}
Subsequently, the second term of \eqref{eq:epsilon_approx} is obtained as
\begin{align} \label{eq:epsilon_approx_pt2}
&\left(\frac{1}{2} + \sqrt{\frac{n (2^R - 1)}{4\pi 2^R}} \right) \cdot \nonumber \\
& \left(e^{\left(- \pi \lambda D^2 \frac{2\pi/\beta}{\sin (\pi/\beta)} \left(B\right)^{2/\beta} \right)} - e^{\left(- \pi \lambda D^2 \frac{2\pi/\beta}{\sin (\pi/\beta)} \left(A\right)^{2/\beta} \right)} \right) \nonumber \\
&- \left( \frac{2}{\beta}\sqrt{\frac{n}{4\pi (2^{2R} - 2^R)}} \pi \lambda D^2 \frac{2\pi/\beta}{\sin (2\pi/\beta)} \right) \cdot \nonumber \\
& \left\{A^{1+\beta/2} E_{-\beta/2}\left(A \pi \lambda D^2 \frac{2\pi/\beta}{\sin (2\pi/\beta)}\right) \right. \nonumber \\
& \left. - B^{1 + \beta/2}E_{-\beta/2}\left(B \pi \lambda D^2 \frac{2\pi/\beta}{\sin (2\pi/\beta)}\right) \right\},
\end{align}
where $E_{n}(z)$ is the exponential integral function defined as
\begin{align}
E_{n}(z) = \int_{1}^{\infty} \frac{e^{-zt}}{t^n} {\rm d} t. 
\end{align}
Combining \eqref{eq:epsilon_approx_pt1} and \eqref{eq:epsilon_approx_pt2} together, we finally have \eqref{eq:epsilon_approx_total}. This completes the proof.
\end{proof}


We verify Proposition \ref{prop:approx} by comparing to an integral expression \eqref{eq:error_prob_int}. As observed in Fig.~\ref{fig:verify}, the gap between \eqref{eq:error_prob_int} and \eqref{eq:epsilon_approx_total} is reasonably small.

\section{Numerical Results}
In this section, we provide performance insights via simulations with our analytical results. The simulation setups are described in each caption. 

In Fig.~\ref{fig:3d_plot}, we illustrate $R_{\varepsilon}$ depending on the density $\lambda$ and the channel coding length $n$. 
On one hand, as $\lambda$ increases, the amount of mutual interference also increases, resulting in that the SIR decreases; thereby $R_{\varepsilon}$ also decreases. 
On the other hand, as $n$ decreases, the rate degradation term in \eqref{eq:r_fbl} increases; thereby $R_{\varepsilon}$ decreases. 
As shown in the figure, the achievable rate $R_{\varepsilon}$ sharply increases when the density $\lambda$ decreases, while increasing $n$ only slightly changes $R_{\varepsilon}$. 
This implies that, considering the network-wise perspective of URLLC, decreasing the interference amount is relatively beneficial compared to increasing the channel coding length. 

\begin{figure}[!t]
\centerline{\resizebox{0.9\columnwidth}{!}{\includegraphics{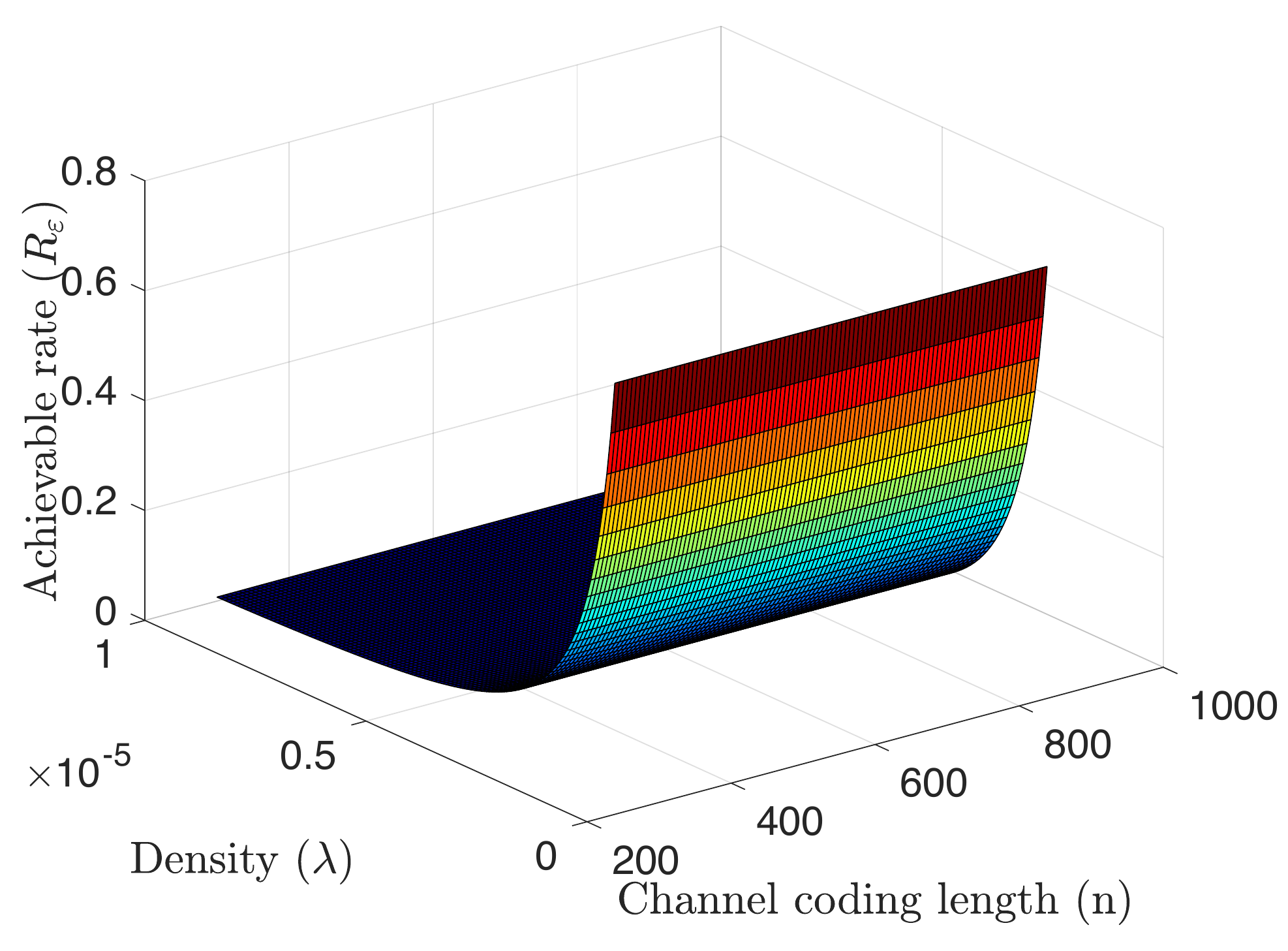}}}     
\caption{The achievable rate $R_{\varepsilon}$ for the decoding error constraint $\varepsilon = 10^{-4}$. The simulation setups are as follows: $\lambda = \{0.5, 1, 2\} \times 10^{-5}$, $n = 500$, $D = 5$, and $\beta = 4$. }
 \label{fig:3d_plot}
\end{figure} 



In Fig.~\ref{fig:fractional}, we illustrate the area spectral efficiency when using fractional frequency reuse. Denoting that the fractional frequency reuse factor as $\eta$, the bandwidth $W$ is separated to $\eta$ bins, and each node in the network is randomly allocated into one of the bins. Then, the effective density of each bin decreases as $\lambda/\eta$ and the channel coding length of each bin also deceases as $n/\eta$. 
In Fig.~\ref{fig:fractional}, we observe that it is beneficial to use fractional frequency reuse. Specifically, the optimal factor $\eta^{\star}$ increases as higher reliability is needed. As $\eta$ increases, the amount of the interference reduces while the channel coding length also decreases. 
The observation implies that when we require highly reliable communication, interference mitigation becomes more important even by sacrificing the channel coding length. 


\begin{figure}[!t]
\centerline{\resizebox{0.9\columnwidth}{!}{\includegraphics{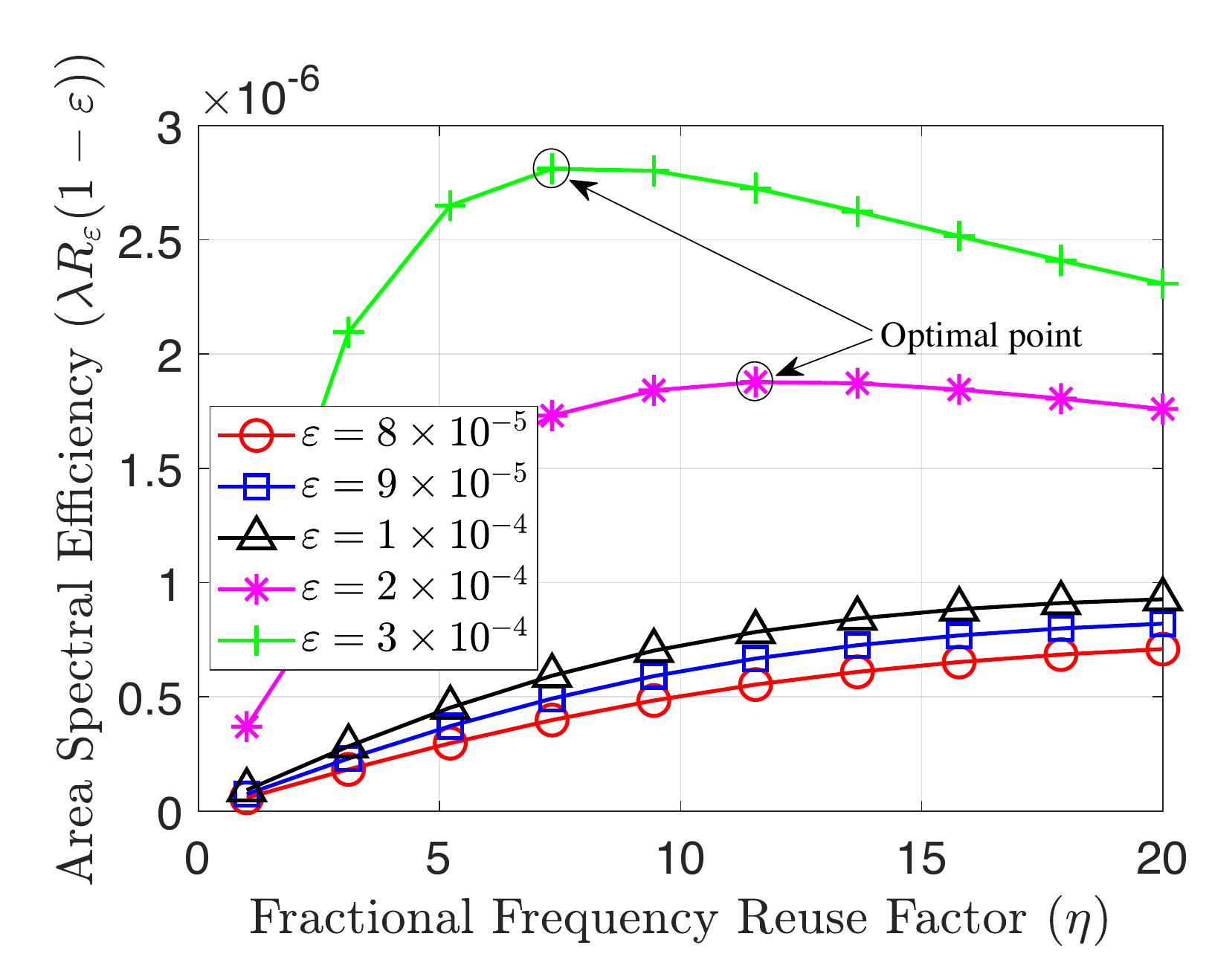}}}     
\caption{The achievable rate $R_{\varepsilon}$ for the decoding error constraint $\varepsilon = 10^{-4}$. The simulation setups are as follows: $\lambda = \{0.5, 1, 2\} \times 10^{-5}$, $n = 500$, $D = 5$, and $\beta = 4$. }
 \label{fig:fractional}
\end{figure}

\section{Conclusions}
In this letter, we have analyzed the achievable rate of URLLC in a random wireless network modeled with a homogeneous PPP. Exploiting the tools of stochastic geometry and finite block-length analysis, we have derived an integral expression of the achievable rate as a function of the decoding error probability, the path-loss exponent, the communication range, the channel coding length, and the density. We also have obtained a tight closed-form approximation. 
Using the analytical results, we have obtained a useful finding that interference mitigation is a key to achieve high rate in URLLC. This encourages to use fractional frequency reuse. 

As future work, it is possible to consider advanced interference mitigation strategies. For example, it is interesting to analyze the performance gains in URLLC by using multiple antennas \cite{park:twc:16}.

\bibliographystyle{IEEEtran}
\bibliography{ppp_shortcode}

\end{document}